%% file: main.tex
\documentclass[11pt, letterpaper]{article}
\usepackage{graphicx,todonotes,subcaption}
\usepackage[margin=1in]{geometry}
\input{style}

\usepackage{color-edits}
\addauthor[VS]{vs}{magenta}

\title{The Time to Consensus in a Blockchain: \\
\large{Insights into Bitcoin's ``6 Blocks Rule''}}
\author{Partha S.~Dey \and Aditya S.~Gopalan \and Vijay G.~Subramanian}
\date{}
\begin{document}
\maketitle

\begin{abstract}
    We investigate the time to consensus in Nakamoto blockchains. Specifically, we consider two competing growth processes, labeled \emph{honest} and \emph{adversarial}, and determine the time after which the honest process permananetly exceeds the adversarial process.
    This is done via queueing techniques.
    The predominant difficulty is that the honest growth process is subject to \emph{random delays}.
    In a stylized Bitcoin model, we compute the Laplace transform for the time to consensus and verify it via simulation.
\end{abstract}

\section{Introduction}
\label{sec:blockchain-prelim}

There has been recent interest in applying blockchains to supply chains~\cite{keskin2024blockchain, iyengar2023economics, iyengar2024blockchain, cui2024supply}, due to the fact that they provide \emph{verifiable consensus} -- that is, it can be determined (at least asymptotically) from the blockchain data structure that all parties agree on a certain subset of the blocks.  For applications such as grocery supply chains, we also need to require that consensus is timely.
However, unlike the widely-studied Bitcoin system, where new blocks arrive significantly more slowly than blocks propagate in the network~\cite{decker2013information}, blocks arrive more quickly in the grocery supply chain setting~\cite{keskin2024blockchain}. This means that the time to consensus is a non-trivial quantity, which we focus on in this paper.

\subsection{The Blockchain Protocol}
In this section, we describe the general form of the blockchain protocol without an adversary.

Consider a set of $N$ nodes on a peer-to-peer network.
Each node creates new binary items, called \emph{blocks}, over time.
Whenever a new block is created, the creating node \emph{references} at least one previous block.
Thereafter, those references are considered a part of the block.
Here, a reference implies that a block creating node \emph{trusts} any referenced block, as well as any indirectly referenced previous blocks.
Usually, there is a fixed rule $f$ that determines how a new block is attached to the blockchain, and we will assume so throughout this document.

The blocks are communicated amongst the nodes via the peer-to-peer network, such that every block is eventually disseminated to each node.
However, the blocks may experience arbitrary communication delays as they are spread through the network.

Thus far, we have addressed the question of achieving consensus on the sequence of blocks in parallel, but we have not addressed the question of identifying a subsequence of blocks that all nodes trust.
For convenience, we label each block by the order of its arrival.
The reference structure means that we can treat the blocks and their references as a directed acyclic graph, which we call a \emph{blockchain DAG}.
Each node maintains its own blockchain DAG, but different nodes' blockchain DAGs may differ due to communication delays.
Each node's blockchain DAG may also differ from the blockchain DAG, which consists of all created blocks.

Henceforth, we assume that vertices are labeled in order of their creation.
Suppose that there is a block labeled $b$ in the blockchain such that for some integer $m_b < \infty$, all blocks with labels at least $(b + m_b)$ have a path to $b$ in the blockchain DAG.
Since $N$ is finite, this implies that each node added a block that has a path to $b$.
We say that the block labeled $b$ is \emph{confirmed}.
The set of such confirmed blocks gives the desired subsequence of blocks trusted by all agents. We make two crucial remarks about this subsequence of confirmed blocks.

\begin{remark}
    The direction of the implication is that if a block is confirmed, then all nodes trust it.
    As a result, the subsequence of confirmed blocks is \emph{not necessarily} a maximal sequence of blocks that are trusted by all nodes.
    There may be other blocks that all nodes trust, but without the confirmation structure, the nodes cannot know of this agreement.
    This point cannot be addressed by the blockchain protocol, and we proceed without further comment.
\end{remark}
\begin{remark}
    It is not obvious or immediate that the set of confirmed blocks is infinite.
    If it is finite, the blockchain protocol takes infinite time to confirm only finitely many blocks (recall that confirmation cannot be determined at any finite time). We, therefore, require additional structure to ensure that the set of confirmed blocks is infinite.
    This additional structure is addressed through a concept called \emph{one-endedness} and is discussed in~\cite{gopalan2020stability, dey2022asymptotic, duffy2023almost}.
\end{remark}

\subsection{Nakamoto Blockchains with Adversaries}
In Bitcoin, new blocks are added to the blockchain according to the \emph{Nakamoto Rule}: a new vertex references one vertex of maximal hop distance to the root.
For the remainder of this work, we assume that a blockchain uses the Nakamoto rule.
For simplicity, we also assume that any ties are broken in favor of the least-indexed block (equivalently, the oldest block).

Our goal is to present a finite analog of block confirmation.
Recall that a block confirmation is a tail event and is not determinable at any finite time.
Nevertheless, if a block $b$ is confirmed, there exists a random variable $M_b < \infty$ such that all blocks with labels at least $(b + M_b)$ have a path to $b$.
Specifically, our goal is to determine the distribution of $M_b$.

In Bitcoin, $M_b$ is often taken to be at most $6$ ``w.h.p.'' (hence the name, ``6 Blocks Rule''), due to an erroneous computation in the Bitcoin whitepaper~\cite{nakamoto2008bitcoin} (hence the quotations around ``w.h.p.'').
In addition, the computation in~\cite{nakamoto2008bitcoin} does not take into consideration the effect of network delays, which are incorporated in this work.
To align the work more closely with the blockchain literature, we consider a model with a worst-case adversary which is based on the one defined by Dembo~\textit{et al.}~\cite{dembo2020everything}.

Our adversary behaves as follows.
Whenever an adversarial node adds a block, we allow them to simultaneously add multiple blocks, according to two rules.
First, if there is any adversarial leaf, one new adversarial block is added that references the adversarial leaf.
Next, if there are honest leaves at greater heights larger than any adversarial leaf, one adversarial block is added per such honest leaf and references that leaf's parent.
It is easy to see that any other choice by an adversary adds a subset of the nodes added by the adversary we consider in this paper.

Our model is the first to consider the dynamics of the time to consensus with an explicit characterization of delay and a worst-case adversary. The model is discussed more precisely in Section~\ref{sec:model}.

\subsection{Contributions of This Work}

Traditionally, Nakamoto-style blockchain security has been interpreted via the lead of the honest blockchain, compared to that of the adversary~\cite{guo2022bitcoin}.
However, the adversary may keep some blocks hidden, and it may be difficult to identify which blocks are actually adversarial.

A slightly different interpretation of blockchain security provides a clean interpretation of our results.
Instead, we consider the security problem from the adversary's perspective: in this case, a worst-case adversary knows the full state of the blockchain, \emph{and} which blocks are adversarial.
Instead of the (state-based) idea of security based on the lead of the honest party, we instead focus on the (time-based) interval of time until an attack fails.
The time-based approach is motivated by the fact that our model is \emph{not} a Markov chain.
Thus, our use of the term \emph{time to consensus} can also be interpreted as the time until a worst-case adversary's attack fails.

More specifically, we pose the time to consensus as the last passage of a $\bZ$-valued random walk.
The last passage of discrete-space processes is a notoriously difficult problem, and in general one cannot obtain the explicit distribution (or its transform).
This is also true in our model, but we are able to capture the behavior of certain other functionals.

In terms of other transaction processing systems, our results can be interpreted as follows.
Noting that blockchains are a data structure for distributed consensus, a guarantee on time to consensus is analogous to the ``3 business day" transaction processing time guarantee for a credit card---e.g., Visa or Mastercard---transaction.
Unlike Visa or Mastercard, the transaction list in a blockchain is \emph{not} privately managed, so it is important to incorporate an adversary in any model.
We note that without an adversary, the time to consensus for our model has already been addressed in a recent paper by Dey and Gopalan~\cite{dey2022asymptotic}.
In emerging applications for blockchains, such as in grocery supply chains~\cite{keskin2024blockchain}, blocks arrive quickly and thus different blocks' propagation through the peer-to-peer network has complicated dynamics~\cite{gopalan2023data}.
As a result, the time to consensus is a non-trivial property of the blockchain dynamics.

Our contributions consist of two types: results for a Bitcoin-specific stylized model, and results for a more general model that handles emerging applications like grocery supply chains.

\subsubsection{Results for Bitcoin}
In Section~\ref{sec:bitcoin}, we use these parameters to create a stylized version of our model to specifically analyze the time to consensus in Bitcoin, before proceeding to our more general model.

For the stylized model, we obtain the exact Laplace transform and tail decay for the time to consensus distribution.
This analysis relies crucially on properties of stable and unstable $M/M/1$ queues and is difficult to replicate in our more general model.
Through simulations, we also examine the time to consensus numerically.

We find that the expected time to consensus is at most $60$ minutes when $p \geq 0.72$, and that the probability that the time to consensus exceeds $60$ minutes is at most 10\% when $p \geq 0.84$, and at most 5\% when $p \geq 0.89$.
These suggest that the commonly taken folk rules require rather conservative estimates on the system parameters, in which the adversary is far from its critical value.

\subsubsection{General Results}
The analysis stylized Bitcoin model relies on knowing the (transforms of) the cycle length of a stable and unstable $M/M/1$ system, conditioned to be finite in the unstable case.
In general, it is difficult to obtain that quantity for more general systems such as $M/G/1$.

Instead, we resort to a different strategy to describe the last passage time.
Our method is indirect, and one challenge that we leave for future work is to revert from the time scale we consider back to the original time scale.
Our approach is to count the number of stable queue cycles until the time to consensus, noting that the positive drift is incurred by ``service completions'' that occur when the queue is empty.
If $X_{(n)}$ is the number of service completions that occur during the $n$-th empty period, and $Y_{(n)}$ is the maximum queue length during the $n$-th busy period, we examine the last passage of the quantity
$$\left(\sum_{i=1}^{n}X_{(n)}\right) - Y_{(n)}.$$

The problem of reverting from the number of cycles to the true time to consensus is difficult because of the positive correlation between large $X$ and $Y$ and large empty and busy periods, respectively.
Nevertheless, the number of cycles until the time to consensus has an important interpretation in the context of blockchains.
One cycle is an interval of time during which the adversary has a lead over the honest parties; our result thus counts the number of such intervals before consensus.

\subsection{Related Literature}
The literature on the time to consensus in a (general) blockchain is sparse, and thus, there are no results for the time to consensus in Nakamoto blockchains under non-trivial network delay.
The closest work is that of Guo and Ren~\cite{guo2022bitcoin}.
The authors use a similar setup as in our work, but specialized to bounded delays.
We note that they do not provide an explicit characterization of the network delay, other than the assumed upper bound for the delay.
Thus, their analysis does not adequately capture the full operating ranges of blockchains, and specifically, does not apply to emerging applications such as grocery supply chains in which network delays could be significant and unbounded.
Due to the simplicity of their model, they are able to reduce the problem to a situation with no delay---specifically, the property they reduce their model to is that all honest blocks are received by all honest nodes before the next honest block is created, which is equivalent to zero delay.
This analysis does not extend to the entire range in which the blockchain is well-behaved in the sense of~\cite{dey2022asymptotic}.
In particular, their technique to reduce the problem to zero delays does not extend to our model.

In the paper by Guo and Ren~\cite{guo2022bitcoin}, the zero-delay property means that the object they study is the last-passage of a skip-free random walk with drift.
Our introduction of non-trivial delays significantly complicates the analysis and requires a different set of tools.
To our knowledge, no other paper considers the appropriate last-passage time.

\section{Stylized Bitcoin Model}
\label{sec:bitcoin}
\subsection{Model}
Let $\xi$ be an improper probability measure supported on $\{1, \infty\}$.
Here, the event $\{\xi = 1\}$ will correspond to zero network delay, while the event ${\xi = \infty}$ corresponds to infinite network delay.
While no data provides $\pr(\xi = 1)$, we will conservatively lower-bound that value in a semi-principled way.
Then, our measure $\xi$ will correspond to a slower network, in terms of propagation delay, than the true network as measured by Decker and Wattenhofer~\cite{decker2013information}.

Owing to the simplicity of this delay model, when there is no adversary, every block will either be added to the longest chain when $\xi = 1$, or it will permanently be a leaf when $\xi = \infty$.

We consider the following discrete-time model.
Independently with probability $p$, each time step is \emph{honest}, and otherwise it is \emph{adversarial}.
Denote by $(\omega_t)_{t \in \bN}$ a sequence of i.i.d. Bernoulli random variables determining whether each time step is honest or adversarial.
Since $\xi$ is supported on $\{1, \infty\}$, it suffices to only keep track of the height of the honest blockchain and the height of the adversarial blockchains.
These evolve as follows:
\begin{align*}
    H_t &= H_{t-1} + \ind_{\omega_t = 1}\ind_{\xi_t = 1}, \\
    A_t &= A_{t-1} + \ind_{\omega_t = 0}.
\end{align*}
We will discuss the initial values $H_0$ and $A_0$ in the next section.

The \emph{time to consensus} is given by $$
\tau_C:=\inf_{t \in \bN}\{t: H_s \geq A_s\ \forall s \geq t\};
$$ 
that is, it is the last passage of the process $H_t - A_t$ to $\bZ_-$.
The time to consensus is \emph{a.s.}~finite when $p\pr(\xi = 1) > 1-p$, and is an improper random variable otherwise.
Thus, we restrict our attention to the case when $p\pr(\xi = 1) > 1-p$: this restriction is often referred to as the \emph{security threshold} in the Bitcoin literature (see, \textit{e.g.},~\cite{gavzi2020tight}).

\subsection{Bitcoin's Time to Consensus} 
Denote by $Q_t := \max(A_t - H_t, -1)$, for $t \geq 0$.
Observe that when embedded into a Poisson Point Process of rate $(1-p) + p\pr(\xi = 1)$, the increments of $Q_t$ can be coupled to those of an $M/M/1$ queue --specifically, taking the arrival rate to be $1-p$ and the service rate to be $p\pr(\xi=1)$, this coupling is achieved via uniformization.
For the remainder of this Bitcoin-specific section, we will use this embedding.
For notational convenience, we will use the following  for the remainder of this section:
\begin{align*}
    \lambda &:= 1-p, &&\mu := p\pr(\xi = 1),\\
    \rho &:= \frac{\lambda}{\mu} = \frac{1-p}{p\pr(\xi = 1)}\in(0,1), &&\hat{p} := \frac{\mu}{\lambda + \mu}\in (1/2,1).
\end{align*}

Due to the coupling with the $M/M/1$ queue, in the stationary regime, we can take a geometric random variable as the initial condition.
Specifically, the stationary measure for $Q_t + 1$ is distributed as $\mathrm{Geom}_0(\rho)$\footnote{For $X\sim\mathrm{Geom}_0(\rho)$, $\pr(X=i)=(1-\rho)^{i}\rho$ for $i\in\mathbb{Z}_+$.}. We start with the following basic facts, which we state without proof~\cite{bertsimas2022queueing}.
\begin{proposition}
    For a stable $M/M/1$ queue with arrival rate $\lambda$ and service rate $\mu > \lambda$, the Laplace transform of the busy period is given by $$B(s) := 
    \frac{\lambda + \mu + s - \sqrt{(\lambda + \mu + s)^2 - 4\lambda\mu}}{2\lambda}.
    $$ 
    Moreover, the Laplace transform of the cycle length and the residual busy period are, respectively, given by 
    \begin{align*}
        \Phi(s):=\frac{\gl}{\gl+s}\cdot B(s), 
    \qquad \Psi(s) := (\mu - \lambda)\cdot \frac{1 - B(s)}{s}.
    \end{align*}
\end{proposition}
\begin{proposition}
    The probability that the busy period (and hence cycle length) of an unstable $M/M/1$ queue with arrival rate $\mu$ and service rate $\lambda < \mu$ is finite is $\rho$.
    The Laplace transform of the cycle length for the unstable $M/M/1$ queue, conditioned on being finite, is given by 
    $$\Gamma(s) := \frac{\mu}{\mu+s}\cdot B(s).
    $$
\end{proposition}
For complex, non-real, $s$, it can be directly checked that the function $B(s)$ is also complex and non-real when $\mathrm{Re}(s) > -s_*$.
For real $s$, the function $B(\cdot)$ is well-defined for $s$ with $\mathrm{Re}(s) \ge -s_\star:= 2\sqrt{\gl\mu}-(\gl+\mu)$ with $B(-s_\star)=\rho^{-1/2}$. The following theorem is an obvious consequence of the previous propositions. 

\begin{theorem}\label{thm:tauC*}
Let $\kappa(s) := \frac{(1-\hat{p})\Gamma(s)}{1-\hat{p}\Phi(s)} = \frac{\mu  (\gl+s) (1-\hat{p})  B(s)}{(\mu+s)(\gl+s-\gl\hat{p} B(s))}$. 
    The Laplace transform of the time to consensus is given by:
    \begin{align*}
        \tau_C^*(s) 
        &:= \left(\rho\Psi(s)+1-\rho\right)\cdot \frac{1-\rho}{1 - \rho \kappa(s)}\\
        &=\left(\rho\Psi(s)+1-\rho\right)\cdot \frac{(1-\rho)(\mu+s)(\gl+s-\gl \hat{p} B(s))}{(\gl+s)(\mu+s)- \gl \hat{p} B(s)\left(\mu(1+\rho^2)+(1+\rho)s\right)}.
    \end{align*}
\end{theorem}

\begin{proof}
    We will discuss the sequence of events until the time to consensus, which follows from a certain mixture over the initial condition. The final result is then obvious from the preceding propositions.

    See the pictorial representation in Figure~\ref{fig:bitcoin-ttc}.

    \begin{figure}[h!]
        \centering
        \includegraphics[width=0.8\linewidth]{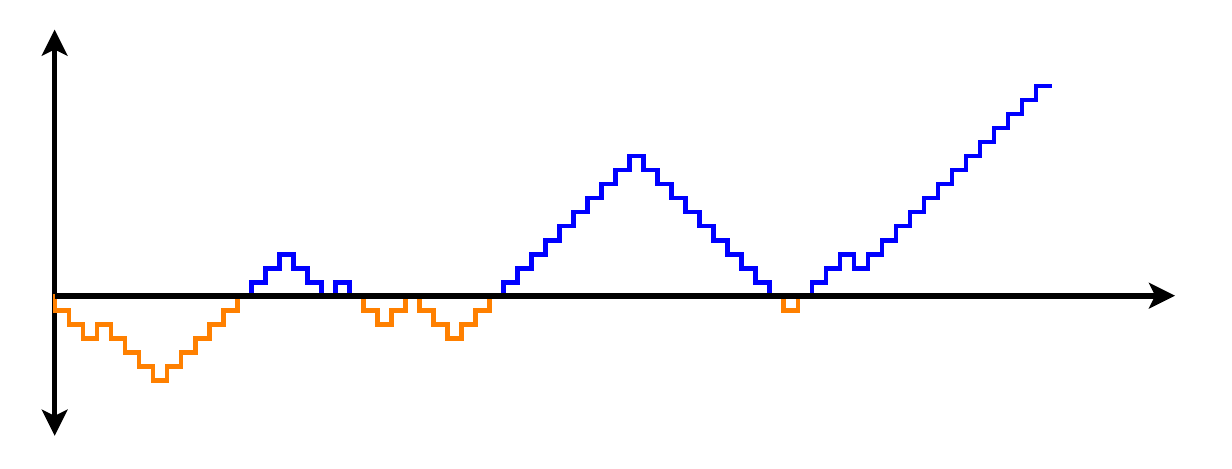}
        \caption{Event leading to the time to consensus. Here, the orange sections are stable busy periods, and the blue sections are unstable ones. The time to consensus occurs at the beginning of the first infintie unstable busy period.}
        \label{fig:bitcoin-ttc}
    \end{figure}

    \textbf{Case 1:} With probability $\rho$, the arriving honest block sees that the adversary is currently ahead of the honest parties. After the residual busy period with Laplace transform $\Psi(\cdot)$, the time to consensus is as follows. There are geometrically many stable cycles with Laplace transform $\Phi(\cdot)$ before a finite unstable cycle with Laplace transform $\Gamma(\cdot)$, and this continues until an infinite unstable busy period occurs.
    Each of these sequences of a finite unstable cycle followed by geometrically many stable cycles has a transform $\kappa(\cdot)$.

    \textbf{Case 2: } With probability $1 - \rho$, the arriving honest block sees that the honest is ahead of the adversary. Then the time to consensus is as follows. There are geometrically many stable cycles before a finite unstable cycles, and this continues until an infinite unstable busy period occurs.
\end{proof}

\subsection{Mean and Tail Decay for Bitcoin's Time to Consensus}

We compute the mean by evaluating $\frac{\mathrm{d}}{\mathrm{d}s}\tau^*_C(s)$ at $s = 0$, to get
\begin{align*}
    \E \tau_C   
    = \frac{\mathrm{d}}{\mathrm{d}s}\tau_C^*(s)\biggl\vert_{s=0}
    =\rho \Psi'(0) + \frac{\rho}{1-\rho}\kappa'(0)
    = \rho \Psi'(0) + \frac{1}{1-\rho} \left(\rho\Gamma'(0) +\Psi'(0)\right).
\end{align*}
Indeed, this expression also follows from Figure~\ref{fig:bitcoin-ttc}.

\begin{theorem}
    There exists $s_{**} \in (0, s_*)$ such that $\tau^*_C(s)$ has a simple dominant pole at $-s_{**}$.
\end{theorem}
\begin{proof}
    Let $$D(s) := (\gl+s)(\mu+s)- \gl \hat{p} B(s)\left(\mu(1+\rho^2)+(1+\rho)s\right)$$ be the denominator in the transform $\tau^*_C(s)$ as given in Theorem~\ref{thm:tauC*}. 
    We first show the existence of a dominant pole at a point $-s_{**}$ in the interval $(-s_*,0)$. 
    Define $\theta:=2\hat{p}(1-\hat{p})\in (0,1/2)$. Note that we have $-s_{*}=(\gl+\mu)(\sqrt{2\theta}-1)$. 
    
    We will work with the function 
    \[
    g(x):=\frac{2}{(\gl+\mu)^2}\cdot D((\gl+\mu)x).
    \]
    Observe that 
    \begin{align*}
        g(x) &=x^2+\theta x+2\theta-1+(1+x-\theta)\sqrt{(1+x)^2-2\theta}\\
        \text{and }
        g'(x) &= 2x+\theta+2\sqrt{(1+x)^2-2\theta} +\frac{\theta(1-x)}{\sqrt{(1+x)^2-2\theta}}.
    \end{align*}
    Clearly, $g(\cdot)$ has finitely many roots.
    Let
    \[
    x_*:=\sup\{x\ge \sqrt{2\theta}-1 \mid g(x)=0\}.
    \]
    Note that, $x_*$ exists and is $<0$, as 
    \begin{align*}
        g(0) &=\sqrt{1-2\theta} \cdot (1-\theta-\sqrt{1-2\theta})>0,\qquad g'(x)>0\text{ for } x\ge 0\\
        \text{and } g(\sqrt{2\theta}-1) &= 3\theta-(2-\theta)\sqrt{2\theta} =(\sqrt{2\theta}-1)(\theta +2\sqrt{2\theta}) <0
    \end{align*}
    with $0 < 2\theta < 1$. 
    By continuity, we have $g(x_*)=0$ and for $x\ge x_*$ we have $g(x)\ge 0$ or $\sqrt{(1+x)^2-2\theta} \ge (1-2\theta-\theta x-x^2)/(x+1-\theta)$. Thus, for $x\in [x_*,0]$ we get
    \begin{align*}
    g'(x) &\ge 2x+\theta+2\cdot \frac{1-2\theta-\theta x-x^2}{x+1-\theta} +\frac{\theta(1-x)}{\sqrt{(1+x)^2-2\theta}}\\
    &= 2-3\theta + \frac{2\theta(1-2\theta)}{x+1-\theta} +\frac{\theta(1-x)}{\sqrt{(1+x)^2-2\theta}} >0.
    \end{align*}
    In fact, the above argument shows that at all roots of $g$ the derivative is strictly positive and thus $x_*$ is the unique solution to $g(x)=0$.
    Define
    \[
    s_{**}=-(\gl+\mu)x_*,
    \]
    so that $D(-s_{**})=0, D'(-s_{**})>0$ and $D(\cdot)$ is strictly positive on $(-s_{**}, \infty)$.

    The existence of a simple dominant pole at $-s_{**}$ follows. 
\end{proof}

The following is a consequence of the Tauberian theorem.
\begin{corollary}
    The tail probability $\pr(\tau^*_C>x)$ decays as $(c+o(1))\cdot e^{- x s_{**}}$ as $x\to\infty$ for some constant $c>0$.
    \label{cor:tail-prob-bitcoin}
\end{corollary}

\subsubsection{Numerical Evaluation}
We choose the probability $\pr(\xi = 1)$ as follows, based on empirical studies.
Decker and Wattenhofer~\cite{decker2013information} find that $95\%$ of blocks are fully propagated within $40$ seconds of their creation; we assume \textit{a priori} that the remaining $5\%$ are such that the next block arrives during their propagation.
Using data from Bowden \textit{et al.}~\cite{bowden2020modeling}, we find that $3.9\%$ of blocks arrive within $40$ seconds of the previous block.
Thus, we make the conservative choice of $\pr(\xi = 1) = 0.9$.
Both of these data sources were taken for the blocks at heights $180,000$ to $190,000$.

We re-scale time such that $\lambda + \mu = \frac{1}{10}\text{ mins}$, which is consistent with the inter-arrival time in Bitcoin (this includes all blocks, honest or adversarial).
We simulate those $p$ for which the expected time to consensus is at most $60$ minutes, which is the expected time for $6$ blocks of any type to arrive in the Bitcoin system.
This value is approximately given by $p \geq 0.72$.
For each $p \in [0.72, 1]$, (with increments of $0.01$), we simulate the system until 1000 blocks arrive with $H_{(\cdot)} > A_{(\cdot)}$ and use that as a proxy to determine the last passage time from the trajectory.
For each such $p$, we simulate $25,000$ independent copies of the system.
The mean time to consensus, and the fraction of samples in which the time to consensus exceeds $60$ minutes, are shown in Figures~\ref{subfig:a} and~\ref{subfig:b}.

\begin{figure}[htbp]
	\begin{subfigure}[t]{.43\columnwidth}
		\centering
		\includegraphics[height=2.5in]{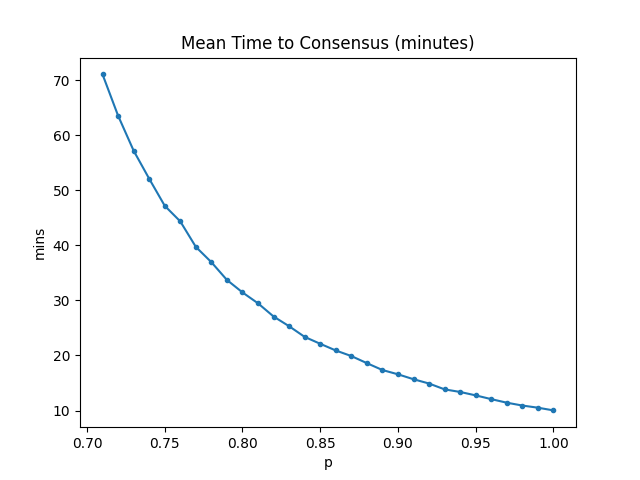}
		\caption{Mean time to consensus.}
		\label{subfig:a}
	\end{subfigure}
	\begin{subfigure}[t]{.57\columnwidth}
		\centering
		\includegraphics[height=2.5in]{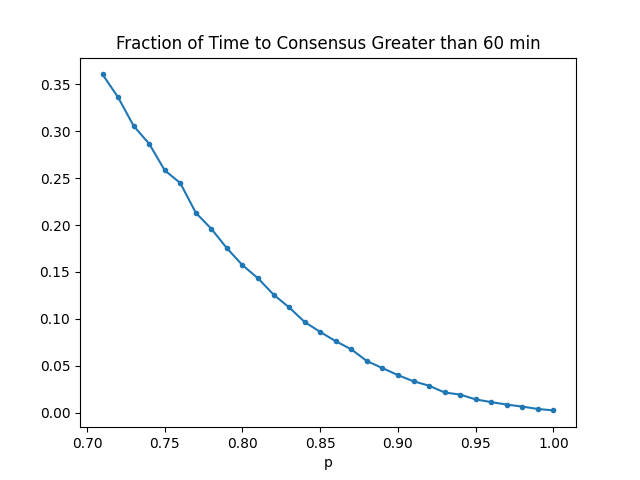}
		\caption{Fraction of samples with \emph{time to consensus} $>60$ minutes.}
		\label{subfig:b}
	\end{subfigure}
	\caption{Simulations in the stylized Bitcoin model with 25,000 samples.}
	\label{fig:TTC}
\end{figure}

For $p = 0.72$ (the closest mean time to consensus to $60$ minutes) and for $p \in \{0.84, 0.89\}$ (the closest values such that the time to consensus only exceeds $60$ minutes 10\% and 5\% of the time, respectively), we also show the empirical distribution of the time to consensus.
These are shown in Figure~\ref{fig:Emp}.
We also plot the theoretical rate predicted in Corollary~\ref{cor:tail-prob-bitcoin} and observe that the slopes match, as expected.

\begin{figure}[htbp]
	\begin{subfigure}[t]{.33\columnwidth}
		\centering
		\includegraphics[width=\linewidth]{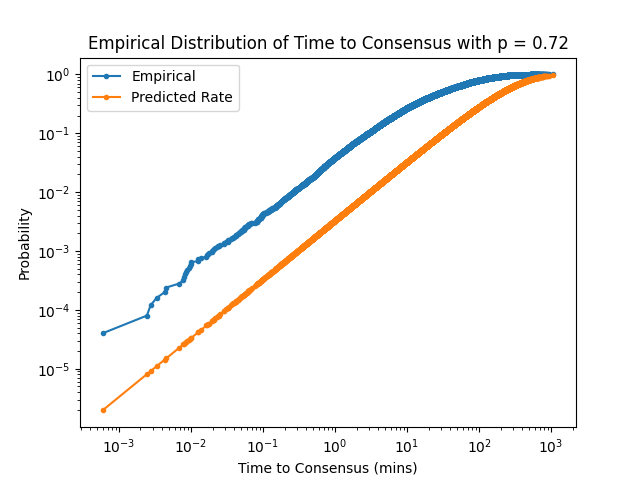}
		\caption{$p=0.72$.}
		\label{subfig:a1}
	\end{subfigure}
	\begin{subfigure}[t]{.33\columnwidth}
		\centering
		\includegraphics[width=\linewidth]{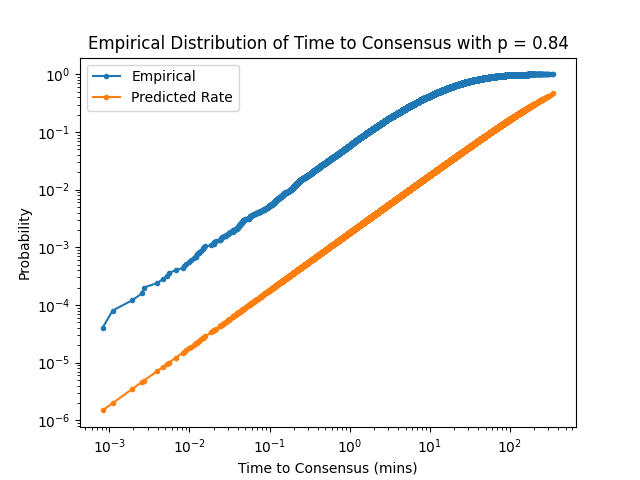}
		\caption{$p=0.84$}
		\label{subfig:b1}
	\end{subfigure}
    \begin{subfigure}[t]{.33\columnwidth}
		\centering
		\includegraphics[width=\linewidth]{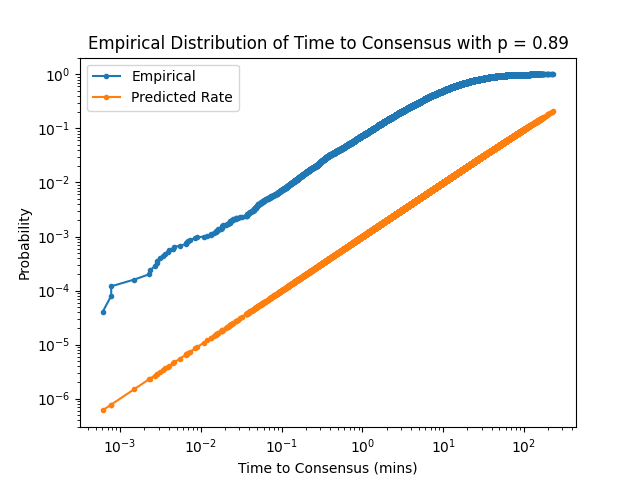}
		\caption{$p=0.89$}
		\label{subfig:c1}
	\end{subfigure}
	\caption{Empirical distribution of \emph{time to consensus} with $p\in\{.72,.84,.89\}.$ The blue line is the empirical distribution, and the orange line is the theoretical rate predicted by Corollary~\ref{cor:tail-prob-bitcoin}}
	\label{fig:Emp}
\end{figure}

\section{General Model and Main Results}
\label{sec:model}
\subsection{Stochastic Growth Process Model}
Our model is as follows.
Fix a parameter $p>{1}/{2}$.
Let $(\omega_t)_{t \in \bN}$ be \emph{i.i.d.}~$\mathrm{Ber}(p)$, and let $(\xi_t)_{t \in \bN}$ be an \emph{i.i.d.}~sequence of $\bN$--valued random variables with finite mean.
The driving sequences $(\omega_t)$ and $(\xi_t)$ are independent.

We consider the following stochastic growth model $(G_t)_{t \in \bZ_+}$. Let $G_0$ be a tree consisting of a single root vertex and no edges. 
The root vertex is both \emph{honest} and \emph{adversarial} and will continue to be the root vertex in $G_t$ for all $t \geq 0$.
Throughout, we refer to the \emph{honest subgraph} and \emph{adversarial subgraph} to be the subgraphs consisting of honest and adversarial vertices, respectively.
The dynamics at each time $t \geq 1$ are as follows.
Whenever a new vertex is added, it has degree one. If the new vertex connects to an existing vertex $v$, we say that the new vertex is added to $v$.

If $\omega_t = 1$, then a single honest vertex is added to the tree.
This is done via the \emph{Nakamoto rule}: the new vertex connects to the least-indexed vertex of maximal hop distance to the root in the honest subgraph of $G_{(t - \xi_t)_+}$.
Notice that the honest subgraph is always a connected tree.

If instead $\omega_t = 0$, then adversarial vertices are added as follows.
For each leaf in the adversarial subgraph, one adversarial vertex is added to that leaf.
In addition, for each vertex in the honest subgraph, if its parent does not already have an adversarial child, one adversarial vertex is added to its parent.
Here, the driving variables $(\xi_t)$ represent \emph{network delays}, which only affect the honest parties.
The assumption that the adversarial parties do not face network delay corresponds to a worst-case setting for blockchain operation.

It has been established by Dey and Gopalan~\cite{dey2022asymptotic} that the limiting honest subgraph is a one-ended tree.
In this limiting tree, there is exactly one confirmed block at any height $h$ from the root (notice that it is the first honest block added at height $h$).
Our question of interest is, for large enough $p$, to determine the distribution of the time after which the height of the honest subgraph permanently exceeds the height of all adversarial descendants of a confirmed block of height $h$.
This is the \emph{time to consensus}.

\subsection{Queueing Structure}
It is shown in Dey and Gopalan~\cite{dey2022asymptotic} that the height of the honest subgraph satisfies the following dynamics:
\begin{align*}
    H_t &= -1 &&\forall\ t \leq 0,\\
    H_0 &= 0, \\
    H_t &= \max\left(H_{t-1}, \omega_t\cdot (1 + H_{t - \xi_t})\right) &&\forall\ t > 0,
\end{align*}
and that the instances upon which $H$ increments form a renewal process on $\bZ_+$ with a point at $0$.
In particular, it is shown that the distribution of the inter-renewal duration $R=R_p$ is given by:
\begin{align}\label{eq:Rp}
    \pr(R_p > r) = \prod_{i=1}^{r}\left(1-p + p\pr(\xi > i)\right) \text{ for } r\ge 0,
\end{align}
and that for any $k \ge 1$, $\E R_p^k < \infty$.
The MGF for $R_p$ is also shown to exist in~\cite{dey2022asymptotic}.

Using this, our dynamics can be considered as a queueing process as follows.
We treat increments of the adversarial longest chain as arrivals, and increments of the honest longest chain as services.
The driving variables $(\omega_t)$ partition the process into arrival points and service points: with probability $p$, a point is a service point, and otherwise it is an arrival point.
Then, the arrivals form an intensity $(1-p)$ Bernoulli point process.
The service is distributed as $R_p$.
Note that \emph{the services are not independent of the arrival process}, due to the fact that the delays faced by honest parties are with respect to \emph{all} vertices, not just the honest ones.

Here, the ``service'' continues even if the queue is empty, similar to the unstoppable server of~\cite{haviv2022queueing, boxma2025priorities}.
Thus, the first customer in a busy period receives a service that is the residual distribution of the true service distribution.
Throughout, we will use standard queueing-theoretic language without defining it.

Letting $A_t$ be the height of the longest adversarial vertex at time $t$, the height of the adversarial path evolves as follows: $A_t=-1$ for $t\le 0$ and 
\begin{align*}
    A_t &= A_{t-1} + 1-\omega_t, \quad t > 0.
\end{align*}
Observe then that the time to consensus (if it exists) is given by the last passage of $H_t-A_t$ to $\bZ_-$, where the dynamics are according to the following Lindley recursion:
\begin{align*}
    A_t &= A_{t-1} + 1 - \omega_t, \\
    H_t &= \max\left(H_{t-1}, \omega_t\cdot (1 + H_{t - \xi_t})\right), \\
    Q_t &= \max(A_t - H_t, -1).
\end{align*}

The reason that we take the maximum with $-1$ in the definition of $Q_t$ is that if a transaction occurs at some height $h$ that the adversary wishes to attack, the adversary must also place a block at height $h$.
The queue has a stationary distribution whenever $(1-p)\E R_p < 1$; stability of the queue is clearly equivalent to the time to consensus being \emph{a.s.}~finite. 

To see this first, we embed the discrete--time process into a unit-intensity Poisson point process. Note that when $(1-p)\E R_p < 1$, the expected number of arrivals during a service is smaller than $1$.
This establishes the existence of the stationary distribution at the epochs of departures, as is standard in the analysis of the $M/G/1$ queue.
Since $Q$ only increments by values of $-1, 0, 1$, it follows that the distribution of the queue length seen by an arriving customer is equal to that left behind by a departing customer.
The final result follows by PASTA.

\subsection{Model Characterization via Queue Cycles and Main Result}
We will characterize the time to consensus in terms of various functionals of the cycles of the queue.
To do this, we denote by the subscript $(n)$ any functional of the $n$-th cycle.
Let $Y_{(n)}$ be the maximum queue length in the $n$-th cycle, and let $X_{(n)}$ be the number of pseudo-service completions in the $n$-th empty period.
Notice that the number of queue cycles required to the time to consensus is invariant to whether or not the vertex in question was at the start of the busy period or not.

Denote by $S_{(n)} := \sum_{i=1}^{n}X_{(i)}$.
We have that $(X_{(i)})_i$ and $(Y_{(i)})_i$ are independent \emph{i.i.d.}~sequences and that $X_{(1)}$ is a geometric random variable; we will compute its parameter later.
It is important to note that $(X_{(i)})_i$ and $(Y_{(i)})_i$ are independent even though the arrivals and services are dependent.
Finally, we consider the last passage of the following process to $\bZ_-$,
$$B_{(n)} := S_{(n)} - Y_{(n)}.$$
A pictorial representation of the queueing cycle description of our dynamics is in Figure~\ref{fig:BRW}.
Our main results are as follows.

\begin{lemma}\label{lem:p_c}
    There is a unique solution $p_c$ in $(0, 1)$ to the equation $(1-p)\E R_{p} = 1$.
    The time to consensus is a.s.~finite when $p > p_c$.
\end{lemma}
\begin{proof}
The lemma follows by showing the monotonicity of $\E R_p$ in $p \in (0, 1)$.
Indeed, if $p > q$, then 
\begin{align*}
    \pr(R_p > r) &= \prod_{i=1}^{r}\biggl((1 - p)\pr(\xi \le i) +\pr(\xi > i)\biggr)\\& \leq \prod_{i=1}^{r}\biggl((1 - q)\pr(\xi \le i) +\pr(\xi > i)\biggr) = \pr(R_q > r),
\end{align*}
for any $r \geq 0$, and the result follows.
\end{proof}

\begin{theorem}\label{thm:cycle-tail-decay}
Assume that the delay~$\xi$ satisfies $\E \xi < \infty$. For $p\in (p_c,1)$, let $z_*=z_*(p)\in ((1-p)/p,1)$ be the unique solution to the equation
\[
    \sum_{r=1}^\infty \prod_{i=1}^r\left(\frac{1-p}{z_*}+p\cdot \pr(\xi>i) \right) = \frac{p}{1-p}.
    \]
    and let $j_0 := \E p^{R_1}$.
    Let $T$ be the time index of the last passage time of $B_{(n)}$ to $\bZ_-$. Then 
    \[
    C_1\cdot  \gc^{t} \leq \pr(T \geq t) \leq C_2\cdot \gc^{t},\qquad t\ge 1
    \]
    for some constants $0<C_1<C_2<\infty$ where 
    \[
    \gc := \frac{1-j_0\ }{1-j_0\cdot z_*}<1. 
    \]
\end{theorem}

The quantity $j_0$ in the statement of Theorem~\ref{thm:cycle-tail-decay} can be interpreted as follows.
By embedding our dynamics into a Poisson process, the arrivals to the queue are Poisson.
Letting $J$ denote the distribution of the number of arrivals during a single service, as is standard in the analysis of $M/G/1$, we have $j_0 = \pr(J = 0)$.

\begin{remark}
    Technically, we do not use the moment assumption on $\xi$ to prove Theorem~\ref{thm:cycle-tail-decay} as $R$ has finite moments of all orders independently of the distribution of $\xi$.
    Nevertheless, that moment assumption is required for the blockchain data structure to be well-behaved, so we require it anyway~\cite{dey2022asymptotic}. When $\xi\equiv 1$, \ie\ there is no delay, we get $p_c=1/2$ and $z_\star=(1-p)/p$.
\end{remark}

\begin{figure}[htbp]
    \centering
    \includegraphics[width=0.8\linewidth]{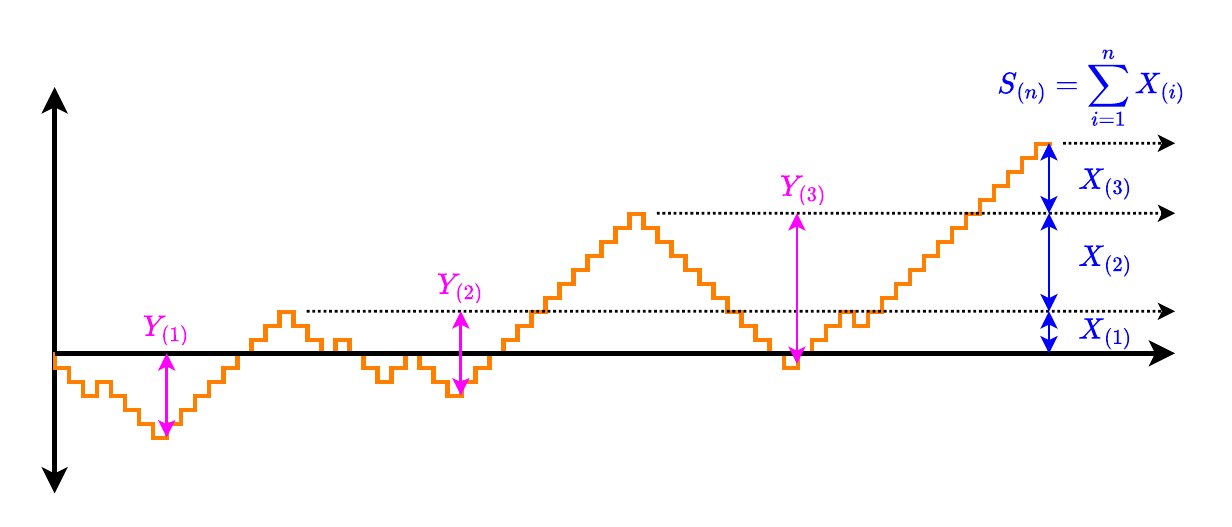}
    \caption{Queueing Cycle Behavior of our Model}
    \label{fig:BRW}
\end{figure}

In this work, we do not explicitly compute the distribution of the time to consensus in the original time scale of the stochastic growth process model.
Instead, we exploit a certain queueing structure within the model and obtain tail bounds on the number of (i.i.d.) cycles of the queue required to reach the time to consensus.
It is \emph{not} an immediate computation to determine properties of the time to consensus in the original (non-queue-cycle) timescale from our result, as conditioned on having passed the time to consensus, the structure of the queue cycles is no longer independent of the number of queue cycles required.

\section{Distributions of $R, J$, $X$, and $Y$}
While the empty ``queue'' in our dynamics corresponds to state $-1$, we will treat an empty queue as having $0$ customers in this section.
Hence, if the adversary's lead is given by $y$, that corresponds to the queue having $y+1$ customers.

By embedding our discrete-time process into a Poisson process, it is easily seen that the queue length process is Markov when sampled at (just after) departure epochs.
Our techniques are similar to the standard analysis of $M/G/1$, and we (re)-define the following:
\begin{itemize}
    \item $R$ is the service time distribution.
    \item $J$ is the distribution of the number of arrivals during a single service.
    \item $X$ is the distribution of the number of ``service completions'' that occur during an empty period.
    \item $Y$ is the distribution of the maximum queue length during a busy period.
\end{itemize}

We will restate the tail distribution of $Y$ at the end of this section in terms of our blockchain process instead of in the translated queueing process.

\subsection{Distributions of $R$ and $J$}
Recall that the service distribution $R_p$ satisfies:
$$
\pr(R_p > r) = \prod_{i=1}^{r}\bigl(1-p+p\pr(\xi> i)\bigr)= \prod_{i=1}^{r}\bigl(1-p\pr(\xi \le i)\bigr) \text{ for } r\ge 0.
$$
In particular, we have the following.
\begin{lemma}
    The random variable $R_p$ satisfies the following stochastic domination relationship:
    \[
\geom(1-p) 
\preccurlyeq R_p 
\preccurlyeq
\geom(1-p) + r_p
\]
where 
\[
r_p:=\lceil p \E(\xi -1)/\abs{(1-p)\log(1-p)}\rceil.
\]
Moreover, $R_p$'s are stochastically decreasing in $p$.
\end{lemma}
\begin{proof}
Clearly,  for all $r\ge 0$, we have
\[
 (1-p)^r \le \pr(R_p > r) \le (1-p)^r \cdot \exp(p/(1-p)\cdot \E(\xi -1)) \le (1-p)^{r-r_p}
\]
and $\pr(R_p>r)\le \pr(R_{p'}>r)$ when $p>p'$. The proof follows. 
\end{proof}

Let $J=J_p$ denote the number of arrivals during a service period. Let 
\[
\theta_i(p):=\frac{(1-p)\cdot \pr(\xi \le i)}{1-p\pr(\xi \le i)} \in [0,1]\text{ for } i\ge 1.
\]
Note that we get 
\[
\E\left(\sum_{i=1}^{ R_p-1} \theta_i(p)\right) = \frac{1-p}{p}.
\]
by Fubini and differentiating w.r.t.~$p$ the relation
\[
\sum_{r=1}^\infty \pr(\xi \le r)\prod_{i=1}^{r-1}\bigl(1-p\pr(\xi \le i)\bigr) = 1/p.
\]
\begin{lemma}
    Conditionally on $R_p=r, r\ge 1$, we have 
\begin{align}
(J \mid R_p = r)\equald \sum_{i=1}^{r-1} (1-\eta_i);
\label{eq:lem4.2RHS}
\end{align}
in particular, we have
\[
R_p-1 - \sum_{i=1}^\infty \eta_i \preccurlyeq J_p \preccurlyeq R_p-1,
\]
where $\eta_i\sim\ber(p(1-\theta_i)), i\ge 1$ are independent of each other and of $R_p$. Moreover, for $z\in (1-p,\infty)$, we have
\[
\E z^{-J_p} = \E (\hat{p}/p)^{-R_{\hat{p}}} <\infty,
\]
for $z>1-p$ where $\hat{p}:=pz/(1-p+pz).$
\end{lemma}
\begin{proof}
    One can think about the effect of the conditioning $(R=r)$ as forcing the following: for the $i$-th variables we have either $\go_i=0$ or $\go_i=1, \xi_i>i$ for all $i=1,2,\ldots,r-1$, and $\go_r=1,\xi_r\le r$. Note that the  RHS of~\eqref{eq:lem4.2RHS} is stochastically increasing in $r$ and 
    \[
    1-p+p\theta_i = \frac{1-p}{1-p\pr(\xi\le i)} = \pr(\go_i=0\mid \go_i=0\text{ or } \go_i=1, \xi_i>i).
    \]

For $z\in (1-p,\infty)$
\begin{align*}
    \E z^{-J_p} &= \E\prod_{i=1}^{R-1} \frac{(1-p)z^{-1} +p\pr(\xi>i)}{1 - p\pr(\xi\le i)}\\
     &= \sum_{r=1}^\infty z^{1-r}\cdot p \pr(\xi \le r)\prod_{i=1}^{r-1}\bigl(1-p+pz\pr(\xi > i)\bigr)\\
     &= \sum_{r=1}^\infty ((1-p+pz)/z)^r\cdot \frac{pz}{1-p+pz} \pr(\xi \le r)\prod_{i=1}^{r-1}\biggl(1-\frac{pz}{1-p+pz}+\frac{pz}{1-p+pz}\cdot \pr(\xi > i)\biggr)\\
     &= \E (\hat{p}/p)^{-R_{\hat{p}}} <\infty,
\end{align*}
where $\hat{p}:=pz/(1-p+pz).$
Note that, for $z=(1-p)/(1-\eps), \eps\in (0,1)$, we have
$\E z^{-J_p}<\infty$ and $\E (1-p)^{-J_p}=\infty.$
\end{proof}

In particular, 
\[
\E(J_p\mid R_p)= (1-p)(R_p-1) + p\sum_{i=1}^{R_p-1}\theta_i,
\]
which implies that $\E J_p=(1-p)\E R_p.$
For stability, we need $(1-p)\E R_p <1$ or $p>p_c$ where $p_c$ is the unique solution of the equation $(1-p)\E R_p =1$.

\begin{lemma}
    For $p\in(p_c,1)$, there exists a unique $z_\star=z_\star(p,\xi)\in ((1-p)/p,1 )$ such that
    \[
    \E z_\star^{-(J_p-1)} =  1.
    \]
    Moreover, $z_*$ satisfies the equation
    \[
    \sum_{r=1}^\infty \prod_{i=1}^r\left(\frac{1-p}{z_*}+p\pr(\xi>i) \right) = \frac{p}{1-p}
    \]
    and satisfies $z_*\le (1-p)/(p\pr(\xi=1)).$
\end{lemma}
\begin{proof}
    The function $\gl\mapsto \log \E\exp(\gl (J_p-1))$ for $0>\gl>\log(1-p)$ is finite, convex, with derivative at zero being $\E(J_p-1)<0$. Thus there is unique solution $\gl_*$ such that $\log \E\exp(\gl_* (J_p-1))=0$. We define $z_*=\exp(-\gl_*)$. Expanding the expectation $\E\exp(\gl_* (J_p-1))$ we get the result
    \[
    \sum_{r=1}^\infty \prod_{i=1}^r\left(\frac{1-p}{z_*}+p\pr(\xi>i) \right) = \frac{p}{1-p}.
    \]
    Rewriting, we get, with $t_*=(1-p)/(pz_*)$, 
    \[
    \sum_{r=1}^\infty (1-p)p^{r-1}\cdot  \prod_{i=1}^r\left(t_*+\pr(\xi>i) \right) = 1.
    \]
    Clearly, $t_*\le 1$, otherwise the lhs will be strictly bigger than $1$.
Thus, $z_\star(p)\ge (1-p)/p$. Similarly, if $q:=\pr(\xi=1)>0$, then $t_* \ge q$. Thus, $z_\star(p)\le \min(1,(1-p)/(pq)).$ 
\end{proof}

\begin{lemma}
    We have for all $n\ge 1$,
    \begin{align*}
       c_p \le \frac{\pr(J\ge n)}{(1-p)^n} \le C_p
    \end{align*}
    for some constants $c_p,C_p$ depending on $p,\xi$.
\end{lemma}
\begin{proof}
    We use the fact that
    \[
    \geom(1-p) -\sum_{i=1}^\infty \eta_i \preccurlyeq R_p -\sum_{i=1}^\infty \eta_i \preccurlyeq J_p+1 \preccurlyeq R_p \preccurlyeq \geom(1-p) +r_p.
    \]
    Thus, 
    \[
    \pr(J_p\ge n) \le \pr(\geom(1-p) \ge  n+1-r_p) \le (1-p)^{n-r_p}
    \]
    and
    \begin{align*}
        \pr(J_p\ge n) &\ge \pr\left(\geom(1-p) -\sum_{i=1}^\infty \eta_i > n \right)\\
        &=(1-p)^n\cdot \E(1-p)^{\sum_{i=1}^\infty \eta_i}\\
        &= (1-p)^n\cdot \prod_{i=1}^\infty \left(1- \frac{p^2\pr(\xi>i)}{1-p\pr(\xi\le i)}\right)
        \ge (1-p)^n\cdot \prod_{i=1}^\infty \left(1- \frac{p^2}{1-p}\cdot \pr(\xi>i)\right).
    \end{align*}
    The infinite product in the last inequality is finite as $\E\xi=\sum_{i=0}^\infty \pr(\xi>i)<\infty.$
\end{proof}

\begin{corollary}
    We have 
    $
    \sup_n\E(z_\star^{-(J_p-n)}\mid J_p\ge n) <\infty.  
    $
\end{corollary}

\subsection{Distributions of $X$ and $Y$}
Recall that $(X_{(n)})$ and $(Y_{(n)})$ are independent i.i.d. sequences.
It is easily seen that $X$ is geometric with success probability $(1-j_0)$.
We have
\begin{align*}
    j_0 
    &= \sum_{r = 1}^{\infty}\pr(J_p = 0 \mid R_p = r)\cdot \pr(R_p = r) 
    = \sum_{r=1}^{\infty}p^{r} \cdot \pr(\xi \leq r)\prod_{i=1}^{r-1}\pr(\xi > i)=\E p^{R_1}.
\end{align*}
We will bound the tail of $Y$ via a martingale argument.
To do this, we use the standard technique for the $M/G/1$ of considering the values of $Q$ at (just after) the times where $H$ increments.
Notationally, we will use $t_n$ to denote the $n$-th time when $H$ increments.

\begin{proposition}
    Let $\tau := \inf_{n \in \bN}\{n: Q_{t_n} \in \{0, y, y+1, \ldots\}\}$.
    Then the stopped process $z_*^{-Q_{t_n \wedge t_\tau}}$ is a martingale.
\end{proposition}
\begin{proof}
    It suffices to show the martingale property inside the set $(0, y)$.
    Let $\cF_{t_n}$ be the natural filtration of $\{Q_{t_n}\}_{n\in \mathbb{N}}$.
    Indeed, if $Q_{t_n} \in (0, y)$, we have
    $
        \E(z_*^{-Q_{t_{n+1}}} \mid \cF_{t_n}) 
        = z_*^{-Q_{t_n}}\E z_*^{-(J-1)} = z_*^{-Q_{t_n}}.
    $
\end{proof}

Let $p_y$ be the probability that the queue length in a busy period reaches at least $y$.
By the optional stopping theorem, we have:
\begin{align*}
    \E\left(z_*^{-(J-1)} \mid J \geq 1\right) 
    &= 1-p_y + p_y\E\left( z_*^{-Q} \mid Q \geq y\right).
\end{align*}
Here, the left-hand side of the equation is the conditional expectation of the queue length given that a busy period has actually started.
We handle the two conditional expectations separately.

\subsubsection{$\E\left(z_*^{-(J-1)} \mid J \geq 1\right)$}
We have that 
\begin{align*}
    \E\left(z_*^{-(J-1)} \mid J \geq 1\right) &= \frac{z_*}{1-j_0}\cdot (\E z_*^{-J} - j_0) 
    = \frac{1-z_*j_0}{1-j_0} = 1 + \frac{j_0}{1-j_0}\cdot  (1-z_*).
\end{align*}

\subsubsection{$\E\left(z_*^{-Q} \mid Q \geq y\right)$}

Since $0 < z_* < 1$, we have that $z_*^{-y} \leq \E\left(z_*^{-Q} \mid Q \geq y\right)$.
We can get an upper bound as follows.
Suppose that at time $t_{\tau - 1}$, we are at some state $0 < s < y$.
\begin{align*}
    \E\left(z_*^{-Q_{t_\tau}} \mid Q \geq y, Q_{t_{\tau -1}} = s\right) 
    &= z_*^{-s}\E\left(z_*^{-J-1} \mid J \geq y-s+1 \right) \\
    &= z_*^{-s}\E\left(z_*^{-J} \mid J \geq y-s \right) \\
    & = z_*^{-y}\E\left(z_*^{-(J - (y-s))} \mid J \geq y - s\right) \\
    &= z_*^{-y}\E(z_*^{-(J - (y-s))} \mid J \geq y-s) \\
    & \leq z_*^{-y}\sup_{n \in \bN}\E(z_*^{-(J-n)} \mid J \geq n) 
    = Cz_*^{-y},
\end{align*}
for some $1 < C < \infty$.

\subsubsection{Tail bounds for $Y$}
We have:
\begin{align*}
    p_y &= \frac{\E(z_*^{-(J-1)} \mid J \geq 1) - 1}{\E(z_*^{-Q} \mid Q \geq y) - 1},
\end{align*}
and thus, from the optional stopping theorem:
\begin{align*}
    \frac{j_0}{1 - j_0}(1 - z_*)\cdot \frac{1}{Cz_*^{-y} - 1} \leq p_y \leq \frac{j_0}{1 - j_0}(1 - z_*)\cdot \frac{1}{z_*^{-y} - 1},
\end{align*}
for some constant $C > 1$.
Thus, for appropriate constants $C_1, C_2$, we have that $$C_1 z_*^{y} \leq p_y \leq C_2 z_*^{y}.$$

Recall that we are actually interested in the probability $p_{y+1}$.
From the preceding analysis, this can be easily absorbed into the constants $C_1$ and $C_2$.

\section{Proof of Theorem~\ref{thm:cycle-tail-decay}: Last Passage Time}
Let $S_{(n)} := \sum_{i=1}^{n}X_{(n)}$.
The \emph{last passage time} $T$ of $S_{(n)}-Y_{(n)}$ to $\bZ_-$ is defined as the a.s.~finite random variable $$T := \max_{t \in \bN}\{t: S_{(t)} - Y_{(t)} \leq 0\}.$$
Notice that
\begin{align*}
    \{T > t^*\} = \bigcup_{t \geq t^*}\{S_{(t)}\le Y_{(t)}\}.
\end{align*}
Thus, using the union bound, we get
\begin{align*}
    \pr(Y_{(t^*)} 
    \geq S_{(t^*)}) \leq \pr(T>t^*)
    \leq \sum_{t\ge t^*} \pr(Y_{(t)} 
    \geq S_{(t)}).
\end{align*}
Define $$h(x):=\pr(Y\ge x).$$ Using independence of $Y_{(t)}$ and $S_{(t)}$, we get
\begin{align*}
    \E h(S_{(t^*)}) \leq  \pr(T > t^*) \leq \sum_{t = t^*}^{\infty}\E h(S_{(t)}).
\end{align*}
Notice that
\begin{align*}
    \{T>t^*\} = \cup_{t=t^*}^\infty\{Y_{(t)} \ge S_{(t)}\} \supseteq \{Y_{(t)} \ge S_{(t)}\},\ \forall t\ge t^*
\end{align*}
so the union bound yields 
\[ 
\pr(T>t^*) \le \sum_{t=t^*}^\infty \pr(Y_{(t)} \ge S_{(t)})=\sum_{t=t^*}^\infty \E h(S_{(t)}),
\]
and we also have 
\[
\pr(T>t^*) \ge \sup_{t\ge t^*}\pr(Y_{(t)} \ge S_{(t)}) = \sup_{t\ge t^*} \E h(S_{(t)}).
\]

We now use the fact that for some constants $C_1, C_2$, $C_1 z_*^{y} \leq h(y) \leq C_2 z_*^{y}$.
For the lower bound, we have:
\begin{align*}
    \E h(S_{(t^*)}) &= \E h\left(\sum_{i=1}^{t^*}X_{(i)}\right)
    \geq C_1 \left(\E z_*^{X_{(1)}}\right)^{t^*}.
\end{align*}
For the upper bound, we have:
\begin{align*}
    \sum_{t = t^*}^{\infty}{\E\left[h(S_{(t)})\right]} &= \sum_{t = t^*}^{\infty}{\E\left[h\left(\sum_{i=1}^{t}X_{(i)}\right)\right]}  \leq C_2 \frac{\left(\E z_*^{X_{(1)}}\right)^{t^*}}{1 - \left(\E z_*^{X_{(1)}}\right)} = C_3\left(\E z_*^{X_{(1)}}\right)^{t^*}.
\end{align*}
Finally, we note that $\E z^{X_{(1)}}=\frac{1-j_0}{1-j_0\cdot z}$ where $j_0=\pr(J_p=0)$. \qed

\subsection{Further Comments}
In this section, we show that the true decay is closer to the union bound estimate for the upper bound in the previous section than to the lower bound.
For fixed $k > 0$, we have:
\begin{align*}
    \{T>t^*\} & = \bigcup_{t=t^*}^\infty\{Y_{(t)} \ge S_{(t)} \}\\
    & = \left(\bigcup_{t=t^*}^{t=t^*+k-1}\{Y_{(t)} \ge S_{(t)} \} \right)\cup \{T>t^*+k\},
\end{align*}
and using the inclusion-exclusion formula
\begin{align*}
    \pr(T>t^*) & = \sum_{l=1}^{k+1} \left((-1)^{l-1} \sum_{\substack{I \subseteq \{1,2,\dotsc,k+1\}\\ |I|=l} }\pr\left(\bigcap_{i\in I} A_i\right)\right),
\end{align*}
where 
\begin{align*}
    A_i = \begin{cases}
    \{Y_{(t^*+i-1)} \ge S_{(t^*+i-1)} \} & \text{if } 1\leq i \leq k \\
    \{T>t^*+k\} & \text{if } i=k+1
    \end{cases}
\end{align*}

Now, the probabilities of the intersections can be ``evaluated" by iterated conditioning. For example, for $j>i$,
\begin{align*}
    \pr(\{Y_{(t^*+i)}\ge S_{(t^*+i)}\} \cap \{Y_{(t^*+j)}\ge S_{(t^*+j)}\}) & = \E \left[\pr(Y_{(t^*+i)}\ge S_{(t^*+i)}) \pr(Y_{(t^*+j)}\ge \tilde{S}_{i,j}+S_{(t^*+i)})\right],
\end{align*}
where $\tilde{S}_{i,j}~\sim S_{(t^*+j)}-S_{(t^*+i)}$ (which is the sum of independent geometric random variables) which is independent of $S_{(t^*+i)}$. Similarly,
\begin{align*}
    & \pr(\{Y_{(t^*+i)}\ge S_{(t^*+i)}\} \cap \{T> k\}) = \\
    & \E\left[ \pr(Y_{(t^*+i)}\ge S_{(t^*+i)}) \pr\left(\bigcup_{t\geq t^*+k} \{Y_{(t)} \geq \tilde{S}_{(i,t)} + S_{(t^*+i)}\}\right) \right],
\end{align*}
where $\tilde{S}_{(i,t)}$ is a partial sum of geometric random variables that are independent of $S_{(t^*+i)}$. 

We will specifically consider the cases $k = 2$ and $k = 3$ to illustrate this iterative conditioning for the lower and upper bounds, respectively.
These are to highlight the structure of our dynamics that arises from cycle to cycle.

Consider the case when $k = 2$.
We have:
\begin{align*}
    \pr(T > t^*) & \geq \sum_{t \geq t^*} \E h(S_{(t)}) - \sum_{t_1, t_2 \geq t^*}\E [h(S_{(t_1)})h(S_{(t_2)})] \\
    & \geq C_4\left(\E z_*^{(X_{(1)}}\right)^{t^*} - C_2^2\sum_{t_1 > t_2 \geq t^*} \E z_*^{2S_{(t_2)} + \tilde{S}_{(t_1, t_2)}} \\
     &= C_4\left(\E z_*^{X_{(1)}}\right)^{t^*} - C_2^2 \sum_{t > t^*}\E z_*^{2S_{(t)}}\sum_{i=0}^{\infty}{\left(\E z_*^{X_{(1)}}\right)^i} \\ 
     & \geq C_4\left(\E z_*^{X_{(1)}}\right)^{t^*} - C_5\left(\E z_*^{2X_{(1)}}\right)^{t^*}.
\end{align*}
Here, the last constant $C_5$ is a result of the fact that $\E z^{X_{(1)}}=\frac{1-j_0}{1-j_0\cdot z}$.
Obviously the leading order term here is $C_4\left(\E z_*^{X_{(1)}}\right)^{t^*}$, which aligns (as previously shown) with the upper bound of $C_3\left(\E z_*^{X_{(1)}}\right)^{t^*}$.

Next, consider the case when $k = 3$.
We have:
\begin{align*}
    \pr(T > t^*) & \leq \sum_{t \geq t^*} \E h(S_{(t)}) - \sum_{t_1, t_2 \geq t^*}\E [h(S_{(t_1)})h(S_{(t_2)})]  + \sum_{t_1, t_2, t_3 \geq t^*}\E [h(S_{(t_1)})h(S_{(t_2)})h(S_{(t_3)})]\\
    & \leq C_3\left(\E z_*^{(X_{(1)}}\right)^{t^*} - C_2^2\sum_{t_1 > t_2 \geq t^*} \E z_*^{2S_{(t_2)} + \tilde{S}_{(t_1, t_2)}} + C_1^2\sum_{t_1 > t_2 > t_3 \geq t^*} \E z_*^{3S_{(t_3)} + \tilde{S}_{(t_1, t_2)} + \tilde{S}_{(t_2, t_3)}}\\
     &= C_3\left(\E z_*^{X_{(1)}}\right)^{t^*} - C_2^2 \sum_{t > t^*}\E z_*^{2S_{(t)}}\sum_{i=0}^{\infty}{\left(\E z_*^{X_{(1)}}\right)^i} +  C_1^2 \sum_{t > t^*}z_*^{2S_{(t)}}\sum_{i=0}^{\infty}{\left(\E z_*^{X_{(1)}}\right)^i}\sum_{j=0}^{\infty}{\left(\E z_*^{X_{(1)}}\right)^j}\\ 
     & \leq C_3\left(\E z_*^{X_{(1)}}\right)^{t^*} - C_6\left(\E z_*^{2X_{(1)}}\right)^{t^*} + C_7\left(\E z_*^{3X_{(1)}}\right)^{t^*},
\end{align*}
where once again the last constant $C_7$ is a result of the fact that $\E z^{X_{(1)}}=\frac{1-j_0}{1-j_0\cdot z}$.
Once again, the leading order term is $C_3\left(\E z_*^{X_{(1)}}\right)^{t^*}$.


\section{Future Work}
In this paper, we compute the tail probability of the time to consensus distribution for a Nakamoto blockchain with a worst-case adversary in a certain sub-sampled time scale.
Ours is the first work to do this in the presence of a non-trivial network delay, which is important for emerging applications of blockchain technology.
There are three pertinent directions for future work, which both require new technical development:
\begin{enumerate}
    \item What is the most likely path to a large time to consensus?
    \item How do we translate from the time scale in our analysis to the original time scale?
    \item How do we extend the results to the more general setting where more than 1 block may arrive per time step?
\end{enumerate}

The first of these directions is a standard large-deviations analysis, although it is complicated by the $Y_{(n)}$ term in the process $S_{(n)} - Y_{(n)}$, which is non-standard.
Traditional large deviations analyses of queues take place \emph{within} a single cycle of a queue; our problem is instead \emph{across} several cycles of a queue.
Understanding the types of events that lead to a large time to consensus will lead to fundamental guidelines for the development of blockchain operating principles.

The second direction is equally important.
Indeed, knowing that the time to consensus required $k$ cycles of the queue imposes a dependence between the first $k$ cycles of the queue and all cycles thereafter.
Therefore, even in expectation, one cannot simply obtain the expected time to consensus via Wald's identity.
This is despite the fact that the cycle-length distribution (or, at least its transform) is readily computable.

The last direction is important for developing new abstractions for blockchains at a less granular timescale, so that other problems such as security could be assessed in the presence of network delay.
We note that last passage for non-skip-free random walks is a notoriously difficult problem; even for our skip-free general model, it is difficult to obtain exact results.

\bibliographystyle{plain}
\bibliography{references}
\end{document}

%% file: style.tex
\usepackage{amsthm}
\usepackage{amsmath}
\usepackage{amsfonts}
\usepackage{amssymb}
\usepackage{dsfont,color,xfrac,enumitem,mathrsfs}
\usepackage{mathtools,tikz-cd}

\usepackage[numbers]{natbib}
\usepackage[colorlinks,citecolor=blue,urlcolor=blue]{hyperref}
\usepackage{graphicx}

\numberwithin{equation}{section}

\newtheorem{theorem}{Theorem}[section]
\newtheorem{lemma}[theorem]{Lemma}
\newtheorem{thm}{Theorem}[section]

\newtheorem{corollary}[theorem]{Corollary}
\newtheorem{proposition}[theorem]{Proposition}

\theoremstyle{remark}

\newtheorem{obs}[thm]{Observation}

\newtheorem{remark}{Remark}

\renewcommand{\le}{\leqslant} 
\renewcommand{\ge}{\geqslant} 
\renewcommand{\leq}{\leqslant} 
\renewcommand{\geq}{\geqslant}

\newcommand{\ind}{\mathds{1}}
 
\newcommand{\eps}{\varepsilon}

\newcommand{\abs}[1]{\left\vert#1\right\vert}

\newcommand{\ie}{\emph{i.e.,}}

\newcommand{\equald}{\stackrel{\mathrm{d}}{=}}

\def\qed{ \hfill $\blacksquare$}

  \let\gc=\gamma  
     \let\gl=\lambda      \let\go=\omega

\newcommand{\cF}{\mathcal{F}}

\newcommand{\bN}{\mathbb{N}}

\newcommand{\bZ}{\mathbb{Z}}        

\DeclareMathOperator{\E}{\mathds{E}}
\DeclareMathOperator{\pr}{\mathds{P}}

\DeclareMathOperator{\geom}{Geom}
\DeclareMathOperator{\ber}{Ber}

